\documentclass{article}  
\usepackage{amsthm}
\usepackage{amsmath,amssymb}

\usepackage[margin=1in]{geometry}
\usepackage{graphicx}
\begin{document}

\newcommand{\bE}{\ensuremath{\mathbf{E}}}
\newtheorem{theorem}{Theorem}[section]
\newtheorem{proposition}[theorem]{Proposition}
\newtheorem{lemma}[theorem]{Lemma}
\newtheorem{corollary}[theorem]{Corollary}
\newtheorem{Definition}[theorem]{Definition}

\title{Lopsidependency in the Moser-Tardos framework: \\ Beyond the Lopsided Lov\'{a}sz Local Lemma}
\author{David G. Harris\thanks{Department of Computer Science, University of Maryland, 
College Park, MD 20742. 
Research supported in part by NSF Awards CNS 1010789 and CCF 1422569.
Email: \texttt{davidgharris29@gmail.com}}}

\date{}

\maketitle

\begin{abstract}

\small\baselineskip=9pt The Lopsided Lov\'{a}sz Local Lemma (LLLL) is a powerful probabilistic principle which has been used in a variety of combinatorial constructions. While this principle began as a general statement about probability spaces, it has recently been transformed into a variety of polynomial-time algorithms. The resampling algorithm of Moser \& Tardos (2010) is the most well-known example of this. A variety of criteria have been shown for the LLLL; the strongest possible criterion was shown by Shearer, and other criteria which are easier to use computationally have been shown by Bissacot et al (2011), Pegden (2014), Kolipaka \& Szegedy (2011), and Kolipaka, Szegedy, Xu (2012).

\baselineskip=9pt  We show a new criterion for the Moser-Tardos algorithm to converge. This criterion is stronger than the LLLL criterion, and in fact can yield better results even than the full Shearer criterion. This is possible because it does not apply in the same generality as the original LLLL; yet, it is strong enough to cover many applications of the LLLL in combinatorics. We show a variety of new bounds and algorithms. A noteworthy application is for $k$-SAT, with bounded occurrences of variables. As shown in Gebauer, Sz\'{a}bo, and Tardos (2011), a $k$-SAT instance in which every variable appears $L \leq \frac{2^{k+1}}{e (k+1)}$ times, is satisfiable. Although this bound is asymptotically tight (in $k$), we improve it to $L \leq \frac{2^{k+1} (1 - 1/k)^k}{k-1} - \frac{2}{k}$ which can be significantly stronger when $k$ is small.

\baselineskip=9pt We introduce a new parallel algorithm for the LLLL. While Moser \& Tardos described a simple parallel algorithm for the Lov\'{a}sz Local Lemma, and described a simple sequential algorithm for a form of the Lopsided Lemma, they were not able to combine the two. Our new algorithm applies in nearly all settings in which the sequential algorithm works --- this includes settings covered by our new stronger LLLL criterion.

\baselineskip=9pt
\end{abstract}

\maketitle

\section{Introduction}
We begin by reviewing background material on the Lov\'{a}sz Local Lemma (LLL), the Resampling Algorithm of Moser \& Tardos to implement it \cite{moser-tardos}, and its generalization the Lopsided Lov\'{a}sz Local Lemma (LLLL). We discuss some strengthened forms of the LLL, such as Shearer's criterion \cite{shearer}. This will set notation which we will use throughout the paper. In Section~\ref{new-contrib-sec}, we will describe the main contribution of this paper, which is a strengthened form of the LLLL.

This is an extended version of a paper which appeared in the Proceedings of the Twenty-sixth annual ACM-SIAM Symposium on Discrete Algorithms.

\subsection{The Lov\'{a}sz Local Lemma and the Moser-Tardos algorithm}
The Lov\'{a}sz Local Lemma (LLL) is a very general probabilistic principle, first introduced in \cite{lll-orig}, for showing that it is possible to avoid a potentially large set $\mathcal B$ of ``bad-events'', as long as the bad-events are not interdependent and are not too likely. We write $m = | \mathcal B |$. One formulation of this principle is the following: Suppose we have a probability space $\Omega$, and we define a \emph{dependency graph} among all the bad-events, such that any bad-event $B \in \mathcal B$ is independent of all the other bad-events except its neighbors in the dependency graph. We use the notation $B \sim B'$ to denote that $B$ and $B'$ are connected in the dependency graph. (The notion of dependency is intuitively clear, although the formal definition is somewhat technical.)

Now suppose there is some weighting function $\mu: \mathcal B \rightarrow [0, \infty)$, with the property
\begin{equation}
\label{lll-eqn}
\forall B \in \mathcal B \qquad \mu(B) \geq P_{\Omega} (B)  \prod_{B' \sim B} (1 + \mu(B'))
\end{equation}
then, in the probability space $\Omega$, there is a positive probability that \emph{none} of the bad-events $B \in \mathcal B$ are true.\footnote{Note that the standard presentation of the LLL uses the parametrization $x(B) = \frac{\mu(B)}{\mu(B)+1}$, but this alternate parametrization will be necessary for later results.}

The LLL is often seen in the simpler ``symmetric'' form: suppose each bad-event $B$ has probability at most $p$; and suppose each bad-event depends on at most $d$ other bad-events. Then (\ref{lll-eqn}) can be simplified to the criterion $e p (d+1) \leq 1$.

This principle has had many applications in combinatorics, for showing the existence of a wide variety of configurations. Unfortunately, typically the probability of avoiding $\mathcal B$ is exponentially small, so this does not yield efficient algorithms. In \cite{moser-tardos}, Moser \& Tardos developed an amazingly simple efficient algorithm for the LLL, as follows: suppose we have a series of variables $X_1, \dots, X_n$; we wish to assign values to these variables. The probability space $\Omega$ assigns each variable \emph{independently}, with $P_{\Omega} (X_i = j) = p_{ij}$.  We also have a set $\mathcal B$ of forbidden configurations of these variables, which we refer to as \emph{bad-events}. For our purposes, it will suffice to consider bad-events which are \emph{atomic}; that is, any bad-event $B$ can be written $B \equiv (X_{i_1} = j_1) \wedge \dots \wedge (X_{i_r} = j_r)$. 
We abuse notation, so that $B$ is identified with the set $\{ (i_1, j_1), \dots, (i_r, j_r) \}$. Thus, for instance, when we write $(i,j) \in B$, we mean that $B$ demands $X_i = j$. 

The Moser-Tardos algorithm (henceforth referred to as MT) can now be described as follows:
\begin{enumerate}
\item[1.] Draw each variable independently from the distribution $\Omega$.
\item[2.] While there is some true bad-event:
\begin{enumerate}
\item[2a.] Choose a true bad-event $B$ arbitrarily.
\item[2b.] Resample all the variables involved in $B$ according to the distribution $\Omega$.
\end{enumerate}
\end{enumerate}
Under the same conditions as (\ref{lll-eqn}), the criterion for the probabilistic LLL,  \cite{moser-tardos} showed that this algorithm terminates quickly. 

\subsection{The Lopsided Lov\'{a}sz Local Lemma}
In \cite{erdos-spencer}, the LLL was generalized by observing that it is not necessary for bad-events to be fully independent. If the bad-events are \emph{positively correlated} in a certain sense, then for the purposes of the LLL this is just as good as independence. (The precise form of the positive correlation is somewhat involved, but we will not need it in this paper so this intuitive definition will suffice.)

 If bad-events $B, B'$ are not positively correlated in this sense, we say that $B, B'$ are \emph{lopsidependent}. One can likewise build a \emph{lopsidependency graph} (also known as a \emph{negative dependency graph}) on the set of bad-events $\mathcal B$. In this case, the LLL criterion (only slightly modified) still applies: we must have
$$
\forall B \in \mathcal B \qquad \mu(B) \geq P_{\Omega} (B) \Bigl[ \mu(B) + \prod_{B' \sim B} (1 + \mu(B')) \Bigr]
$$

This generalized form of LLL, referred to as the \emph{Lopsided Lov\'{a}sz Local Lemma} (LLLL), has been used in a variety of contexts. A variety of probability spaces fit into this framework, for example random permutations ~\cite{lu-szekeley}, Hamiltonian cycles ~\cite{albert}, and matchings on the complete graph ~\cite{lu-szekeley2}. Only a few applications of the LLLL have corresponding efficient algorithms; for example, \cite{lllperm} gives an MT variant for random permutations and \cite{achlioptas} gives algorithms for other spaces such as Hamiltonian cycles.

One important and simple setting for the LLLL is covered by the original MT algorithm, and this was already described in the original paper of Moser \& Tardos: suppose as above that $\Omega$ chooses each variable independently and the bad-events are atomic. Given two such bad-events $B, B'$, we say that $B, B'$ \emph{agree} on variable $i$ if there is some $j$ with $(i, j) \in B, (i,j) \in B'$. We say that $B, B'$ \emph{disagree} on variable $i$ if there are $j \neq j'$ with $(i, j) \in B, (i, j') \in B'$. Now the relation of disagreeing on some variable defines a lopsidependency graph:
$$
B \sim B' \qquad \text{if $\exists (i,j) \in B, (i, j') \in B', j \neq j'$}
$$

We use the notation $(i,j) \sim (i', j')$ iff $i = i', j \neq j'$. Some related notations will be to write $(i,j) \sim B$ iff there is some $j' \neq j$ with $(i,j') \in B$, and to write $i \sim B$ iff there is some $(i,j) \in B$.  We note that if $B \sim B'$, then $B$ and $B'$ are mutually exclusive events.

We refer to this setting as the ``variable-assignment LLLL,'' since we are independently assigning values to each variable. The $k$-SAT problem is a canonical example; we discuss how the LLLL applies in Section~\ref{sat-section}. In this problem, we are given a collection of clauses in $n$ variables, each involving $k$ distinct literals. Each variables appears in at most $L$ clauses (either positively or negatively). Our goal is to find a solution which makes all the clauses true. It turns out the worst case is when each variable appears $L/2$ times positively and $L/2$ times negatively (see Section~\ref{sat-section} for more details). In this case, we assign each variable to be true or false with probability $1/2$. For each clause, we have a bad-event that the clause is falsified. Now consider a bad-event $B$; it has probability $2^{-k}$. In the LLLL setting, $B$ depends only on clauses which \emph{disagree} with the variables in that clause; as each clause appears $L/2$ times with each polarity, there are $k L/2$ other bad-events which are lopsidependent with $B$. Thus the symmetric LLLL gives the bound $L \leq \frac{2^{k+1} - 2 e}{e k}$
to guarantee that a solution exists (and MT finds it). As we will see in Section~\ref{sat-section}, more careful calculations and more precise forms of the LLL can give slightly better bounds.

\subsection{Shearer's criterion and stronger formulations of the LLL and LLLL}
The LLL and LLLL criteria depend solely on two factors: the probabilities of the bad-events, and the shape of the dependency (or lopsidependency) graph between them. The precise nature of the dependency does not enter into them; for example, we have the same formula whether we are dealing with the variable-assignment setting, or permutations, or matchings, and so on. 

One can ask what is the strongest possible criterion that can be given in terms of these two quantities. Thus, given a dependency graph $G$ and probabilities $P_{\Omega}(B)$, is there guaranteed to be a non-zero probability of avoiding all bad-events? In fact, the exactly tight criterion was shown by Shearer \cite{shearer}, and it is stronger than (\ref{lll-eqn}). Furthermore, \cite{kolipaka} showed that, if the MT algorithm has a small slack compared to this optimal criterion, then it too will converge. 

Unfortunately, the criterion of \cite{shearer} is not easy to compute. A variety of other criteria, which are slightly weaker than \cite{shearer} but easier to apply, have been shown, e.g. \cite{bissacot}, \cite{pegden}, \cite{kolipaka2}. These criteria apply both to the probabilistic form of the LLLL as well as to the MT algorithm. We emphasize that these alternate criteria are all weaker than and are implied by Shearer's criterion. We will discuss them in greater detail in Section~\ref{comparison-sec}.

As noted by \cite{kolipaka}, the MT algorithm applies to a more restrictive model than the LLL, so Shearer's LLL criterion is not necessarily tight in this restrictive class. \cite{kolipaka} gave some toy examples of this situation. Notwithstanding this, most researchers have considered Shearer's criterion to be the ultimate form of the LLL. Other forms of the LLL and algorithms such as MT are attempts to match this bound.

\subsection{Our contributions: a new LLLL criterion}
\label{new-contrib-sec}
In this paper, we do not change the MT algorithm in any way. However, we give a alternate criterion for it to converge. In our opinion, the surprising fact is that this new criterion can go beyond the Shearer criterion. In the LLLL framework, lopsidependency is as good as pure independence; we show that if bad-events agree on a variable, this gives \emph{better bounds} than if they were independent! 

 We reiterate that the Shearer criterion is the strongest possible criterion that can be given \emph{for the level of generality to which it applies}. Our new criterion depends in a fundamental way on the decomposition of bad-events into variables; it cannot be stated in the language of probability and dependency graphs. 

We state our new criterion as follows. 
\begin{Definition}[Orderability] Given an event $E$, we say that a set of bad-events $Y \subseteq \mathcal B$ is \emph{orderable} to $E$, if either of the conditions hold:
\begin{enumerate}
\item[(O1)] $Y = \{ E \}$, or
\item[(O2)] there is some ordering $Y = \{B_1, \dots, B_s \}$, with the following property. For each $i = 1, \dots, s$, there is some $z_i \in E$ such that $z_i \sim B_i,  z_i \not \sim B_1, \dots, z_i \not \sim B_{i-1}$.
\end{enumerate}

Note that $\emptyset$ is orderable to $E$, as indeed it satisfies condition (O2). Recall that when we write $z \sim B$, we mean that there exists $z' \in B$ with $z \sim z'$.
\end{Definition}

\begin{theorem}
\label{var-assignment-thm}
In the variable-assignment setting, suppose there is $\mu: \mathcal B \rightarrow [0, \infty)$ satisfying the following condition:
$$
\forall B \in \mathcal B, \mu(B) \geq P_{\Omega}(B) \sum_{\substack{\text{$Y$ orderable}\\\text{to $B$}}} \prod_{B' \in Y} \mu(B')
$$
then the MT terminates with probability 1. The expected number of resamplings of a bad-event is at most $\mu(B)$. 
\end{theorem}

The LLLL cannot guarantee under these conditions that a satisfactory configuration even exists; for this reason, we view this criterion as going beyond the LLLL. This criterion is about as easy to work with as the original MT criterion --- in some cases, in fact, it can yield significantly simpler calculations. In Section~\ref{comparison-sec}, we compare this to other LLLL criteria.

In Section~\ref{app-section}, we give some applications of this new criterion. We summarize the most important ones here:

\begin{enumerate}

\item{\bf SAT with bounded variable occurrences} Consider the following problem: we have a SAT instance, in which each clause contains $k$ distinct variables. We are also guaranteed that each variable occurs in at most $L$ clauses, either positively or negatively. How can large can $L$ be so as to guarantee the existence of a solution to the SAT instance? 

As shown in \cite{gst}, the LLLL gives an asymptotically tight bound for this problem, namely $L \leq \frac{2^{k+1}}{e (k+1)}$. However, there still is room for improvement, especially when $k$ is small.  As $L$ is growing exponentially, it is arguably the case that large $k$ is not algorithmically relevant anyway. We are able to improve on \cite{gst} to show that when
$$
L \leq \frac{2^{k+1} (1 - 1/k)^k}{k-1} - \frac{2}{k}
$$
then the SAT instance is satisfiable, and the MT algorithm finds a satisfying occurrence in polynomial time. This is always better than the bound of \cite{gst}, and when $k$ is small the improvement can be substantial.

\item {\bf Hypergraph coloring}. Suppose we are given a $k$-uniform hypergraph, in which each vertex participates in at most $L$ edges. We wish to $c$-color the vertices, so that no edge is monochromatic (all vertices receiving the same color). This problem was in fact the inspiration for the original LLL \cite{lll-orig}. There are many types of graphs and parameters for which better bounds are known, but the LLL gives very simple constructions and also provides the strongest bounds in some cases (particularly when $c, k$ are fixed small integers). Strangely, depending on whether $c$ or $k$ is large, one can obtain better bounds using the standard LLL or the LLLL. Thus one can show the bounds:
\begin{equation}
\label{hypergraph-crit}
L \leq \frac{c^k}{k} \max(\frac{(1 - 1/k)^{k-1}}{c} , \frac{1}{(c-1) e} )
\end{equation}

Our approach gives the simpler and stronger criterion:
$$
L \leq \frac{c^k (1-1/k)^{k-1}
   }{k (c-1)}.
$$

Our new criterion is always better than (\ref{hypergraph-crit}), interpolating smoothly between the regimes when $c$ or $k$ is large. This illustrates an advantage of our technique --- despite the daunting form of our new LLLL criterion, in practice it typically gives formulas which are more computationally tractable.

\end{enumerate}

\textbf{Comparison of Shearer criterion to Moser-Tardos.} It is challenging to directly compare the Shearer criterion with the Moser-Tardos algorithm, because they apply to such different contexts: generic probability spaces in the former case and variable configurations in the latter. We defer this discussion to a forthcoming paper \cite{harris3} which shows that the analysis of \cite{gst} (based on the asymmetric LLL) cannot be much improved by using a stronger form of the LLL. In particular, our proof directly based on our new MT criterion is stronger than would be possible from Shearer's criterion, let alone the LLL.

\subsection{Our contribution: a new parallel algorithm}
The original MT framework had a simple parallel RNC algorithm for the Lov\'{a}sz Local Lemma. Frustratingly, although the sequential MT algorithm applied to the variable-assignment LLLL setting, the parallel algorithm did not. 

We remedy this situation in Section~\ref{parallel-mt} by introducing a new parallel randomized algorithm for the variable-assignment LLLL, which requires only a multiplicative slack compared to our new criterion for the sequential algorithm.  (Showing that this algorithm is compatible with the new LLLL criterion requires non-trivial arguments).
\begin{theorem}
Suppose there is $\mu: \mathcal B \rightarrow [0, \infty)$ satisfying the following condition:
$$
\forall B \in \mathcal B, \mu(B) \geq (1 + \epsilon) P_{\Omega}(B) \sum_{\substack{\text{$Y$ orderable}\\\text{to $B$}}} \prod_{B' \in Y} \mu(B')
$$
then our new parallel algorithm algorithm terminates with probability 1. Suppose that the size of each bad-event is at most $M$. Then our parallel algorithm terminates in time $\epsilon^{-1} M (\log \sum_{B \in \mathcal B} \mu(B)) (\log^{O(1)} n) (M + \log^{O(1)} m )$ and $(n m)^{O(1)}$ processors with high probability. (Typically $\sum_{B \in \mathcal B} \mu(B) \leq O(m)$).
\end{theorem}

We list a few applications of these new parallel algorithms:
\begin{enumerate}

\item{\bf SAT with bounded variable occurrences} We have a SAT instance, in which each clause contains at least $k$ variables. We are also guaranteed that each variable occurs in at most $L$ clauses. Then, under the condition $L \leq \frac{2^{k+1} (1 - 1/k)^k}{(k-1)(1+\epsilon)} - \frac{2}{k}$
the parallel MT algorithms find a satisfying assignment in time $\frac{k^{O(1)} \log^{O(1)} n}{\epsilon}$.

\item {\bf Hypergraph coloring}. Suppose we are given a $k$-uniform hypergraph, in which each vertex participates in at most $L$ edges. We wish to $c$-color the vertices, so that no edge is monochromatic. Then, under the condition  $L \leq \frac{c^k (1-1/k)^{k-1}
   }{(1+\epsilon) (c-1) k}$
the parallel MT algorithm finds a good coloring in time $\frac{k^{O(1)} \log^{O(1)} n}{\epsilon}$.

\end{enumerate}

\section{The variable-assignment LLLL}
\label{new-mt-proof}

As we have said, we do not change the MT algorithm in any way. The only change is to the analysis. The analysis of \cite{moser-tardos} is based on two main idea: a \emph{resampling table}, and  building \emph{witness trees} for each resampling that occurs during a run of the algorithm. A witness tree lists the full history of the variables involved in a resampling ---  ``why'' a given resampling occurred. \cite{moser-tardos} describes this in much greater detail, and we recommend that the reader should read that paper for a careful and thorough explanation of the witness tree analysis.

The idea of the resampling table is that, at the very beginning of the algorithm, you draw an infinite list of all the future values for each variables. Then resampling table entry $R(i,j)$ gives the $j^{\text{th}}$ value for each variable $X_i$. Initially, you set $X_i = R(i,1)$; when you need to resample $i$, you set $X_i = R(i,2)$, and so forth. After drawing $R$, the remainder of the MT algorithm becomes deterministic. One can determine, for each witness tree, necessary conditions to hold on $R$.

In the MT algorithm, the choice of which bad-event to resample can be arbitrary, and can even be under the control of an adversary. In fact, MT shows something even stronger than this: even if the choice of which bad-event to resample depends on $R$, then the MT algorithm must still converge with high probability. Thus, the LLLL criterion is strong enough to show convergence even if the resampling is determined by a \emph{clairvoyant adversary}.

Our analysis is also based on witness trees, but it dispenses with the resampling table. Instead, variables are assumed to be resampled in an \emph{on-line} fashion. The choice of which bad-event to resample can be arbitrary, but must depend \emph{solely on the prior state (not the future state) of the system.} 

Because we impose this restriction on the resampling rule, we can use stochasticity to analyze our witness trees, in a way which is not possible with the MT framework. The basic idea is that, whenever we resample some bad-event $B$, the distribution for the new values for its variables is the same as the law of $\Omega$, \emph{even conditional on all prior state.} This is simply not true in the MT framework: the choice of resampling rule could produce a dependency on the future.

We will eventually take a union-bound over witness trees, so it is critical to prune the space of witness trees as much as possible. In other words, we will need the most succinct possible explanation of each resampling. Let us consider a simple example of how we can use our stronger stochasticity assumption to analyze more succinct witness trees. Suppose that we have some bad-event $B \ni (i,j)$, and we want to explain why we eventually resampled $X_i = j$. Suppose there were two earlier events $B_1, B_2$ which included $(i,j')$. These events would be placed as children of $B$ in the standard MT witness tree. This means that we encounter $B_1, B_2$ (in an unspecified order), and when we encounter the second one we select $X_i = j$. 

We can view the witness tree for $B$ as making a prediction: namely, that at the appropriate time, when we choose to resample variable $X_i$, then we set $X_i = j$. If we can fix a specific time at which this prediction should hold, then we can bound its probability. For this purpose, it suffices to only record information about the \emph{later} of the two events $B_1, B_2$. We can discard the information about the earlier resampling. By only retaining the latest occurrence of each variable, we still have all the information we need to deduce the resamplings.  The stochasticity now tells us that whenever the resample the latter of $B_1, B_2$, the new values for the variables in it must have the same distribution as in $\Omega$. The reason that this is true is that the choice of whether to resample $B_1$ or $B_2$ first \emph{cannot depend on the new values of the variables.} 

We see that we have ``compressed'' the information relevant of $B$. This significantly prunes the space of witness trees, but we will have to work much harder to show that it is sufficient.

\subsection{Forming witness trees}
When building witness trees, we will maintain the following key invariant: for any node $v$ in the tree labeled by $B$, the children of $v$ receive distinct labels $B_1, \dots, B_s$ such that $\{B_1, \dots, B_s \}$ is an orderable set for $B$.  This is the key principle behind our new criterion.

We now describe how to form a witness tree for an event of interest $E$. Suppose we have listed, in order, all the bad-events that were ever resampled during MT, listed as $B_1, \dots, B_T$. This is referred to as the \emph{execution log}. Suppose $E = B_t$. We start with the resampled event $E$ at the root. Starting at time $t-1$, we proceed backward through the execution log. For each bad-event $B$ encountered, we see if there is some node $v \in \tau$ for which $B$ is eligible. If so, we add $B$ to the \emph{deepest} such position (breaking ties arbitrarily). We give the following more precise definition of eligibility:
\begin{Definition}[Eligibility]
Suppose we have formed a (partial) witness tree $\tau$, and we have a node $v \in \tau$ labeled by $B$. Suppose the children of $v$ receive distinct labels $B_1, \dots, B_s$. Then we say a bad-event $B'$ is \emph{eligible} for $v$ if $B' \neq B_1, \dots, B_s$ and if $\{B_1, \dots, B_s, B' \}$ is orderable for $B$.
\end{Definition}

We distinguish between two related notions of the witness tree. Let $\hat \tau^{t}$ denote the witness tree corresponding to the resampling at time $t$ during an execution of the algorithm; this is a random variable. We also denote by $\hat \tau^{t_1}_{t_0}$ the witness tree produced in this way, in which we only keep track of events after time $t_0$ (that is, this witness tree only records events between times $t_0$ and $t_1$ inclusive). If $t_0 > t_1$, then $\hat \tau^{t_1}_{t_0}$ is defined to be the null tree. We will sometimes omit the superscript to simplify the notation. By definition $\hat \tau^t_{1} = \hat \tau^t$. 

We also sometimes may wish to discuss a certain labeled tree, and under what conditions it could have been produced. We use then the notation $\tau$ to denote a witness tree in this sense.

One simple definition we will use often:
\begin{Definition}
Consider any variable $i$, and consider a tree $\tau$ with a node $v$. We say that $v$ \emph{involves} $i$, if $v$ is labeled by some bad-event $B$, and $(i,j) \in B$ for some $j$.
\end{Definition}

We list some easy properties of the witness trees produced in this manner:
\begin{proposition}
\label{tree-properties-prop}
\begin{enumerate}
\item Consider any bad-event $B$. Consider the leaf nodes of $\hat \tau$ labeled by $B$; all such nodes must have \emph{distinct} depths in the tree.
\item Consider any variable $i$, and, among all the nodes $v \in \hat \tau$ involving $i$, consider the set of such nodes which are greatest depth in the tree. While it is possible that there are multiple such nodes $v_1, \dots, v_r$, all such nodes must be labeled by $B_1, \dots, B_r$ which \emph{agree} on variable $i$.
\end{enumerate}
\end{proposition}
\begin{proof}
The earlier bad-event would have been eligible to be a child of the later bad-event, and hence would have been placed either there or deeper in the tree.
\end{proof}

In light of Proposition~\ref{tree-properties-prop}, we may define the \emph{active value} for each variable:
\begin{Definition}[The active value of a variable]
Consider any variable $i$, and, among the set of nodes $v \in \tau$ involving $i$, consider the nodes at greatest depth in the tree. All such nodes contain $(i,j)$ for some common value $j$. We denote by $A_i(\tau)$, the \emph{active value} of variable $i$, by this common value $j$. 

If variable $i$ does not appear in $\tau$, we define $A_i(\tau) = \top$, the sure value. By convention, we use $X_i = \top$ as a shorthand for the sure event (the entire probability space). For example, $P_{\Omega}(X_i = \top) = 1$.
\end{Definition}

We note that these types of witness trees look very different from the standard MT construction. For example, the layers in the tree (and even the children of a common parent) do not necessarily form an independent set; there can be multiple copies of a single bad-event in a given layer.

Suppose we are given a tree $\tau$ and a time $t_1$; we want to estimate the probability that $\hat \tau^{t_1} = \tau$. This is the key to the MT proof strategy. We can imagine running the MT algorithm and see whether, so far, it appears that it is still possible for $\hat \tau^{t_1} = \tau$. This is a kind of dynamic process, in which we see what conditions are still imposed in order to achieve this tree. One key point in our rule for forming witness trees is that, as we run the MT algorithm, we will be able to deduce not just $\hat \tau^{t_1}$ but also $\hat \tau_{t_0}^{t_1}$ for all $t_0 \geq 1$. 

We will often omit the superscript $t_1$ in the following; it should be understood.

\begin{proposition}
\label{update-tree-prop}
Suppose we are given the partial witness tree $\hat \tau_{t}$, and we encounter a bad-event $B$ at time $t$. Then $\hat \tau_{t+1}$ is uniquely determined, according to the following rule: if there is a leaf node labeled by $B$, select the deepest such leaf node $v$ (by Proposition~\ref{tree-properties-prop} it is unique) and we have $\hat \tau_{t+1} = \hat \tau_{t} - v$. Otherwise we have $\hat \tau_{t+1} = \hat \tau_{t}$.
\end{proposition}
\begin{proof}
First, suppose that $\hat \tau_t$ did contain such a node $v$. It must be that $v \not \in \hat \tau_{t+1}$. For, if so, then when forming $\hat \tau_{t}$ from $\hat \tau_{t+1}$, we would have placed $B$ as a child of $v$; that is, $\hat \tau_{t}$ would include an additional copy of $B$. So $\hat \tau_{t+1}$ is missing the node $v$ from $\hat \tau_{t}$. As each time step can only affect a single node in the witness tree, it must be that $\hat \tau_{t+1} = \hat \tau_t - v$.

Second, suppose that $\hat \tau_t$ contained no such node $v$. When forming $\hat \tau_{t}$ from $\hat \tau_{t+1}$, we either make no changes or add a single node labeled by $B$. In the latter case, $\hat \tau_{t}$ would contain a leaf node labeled by $B$, which has not occurred. Hence it must be that $\hat \tau_{t} = \hat \tau_{t+1}$ as claimed.
\end{proof}

The other key point is that, from the partial tree $\hat \tau_{t}$, we can deduce some information about the variables:
\begin{proposition}
Consider any variable $i$. At time $t$ of the MT algorithm, we must have $X_i = A_i(\hat \tau_t)$.
\end{proposition}
\begin{proof}
Suppose $B$ is a node of greatest depth containing variable $i$, and we have $(i,j) \in B$, where $j = A_i(\hat \tau_t)$. Suppose $X_i = j' \neq j$ at time $t$.

In order to include $B$ in the witness tree $\hat \tau$, we must eventually resample $B$, which implies that eventually we must have $X_i = j$. As $X_i = j'$ at time $t$, this implies that we must first encounter some bad-event $B' \ni (i,j')$. But then $B'$ would be eligible to be placed as a child of $B$, and so would be placed there or lower. This contradicts that $B$ is the greatest-depth occurrence of variable $i$.
\end{proof}

These propositions together allow us to prove the Witness Tree Lemma:
\begin{lemma}[Witness Tree Lemma]
\label{witness-tree-lemma}
Let $\tau$ be a witness tree with nodes labeled $B_1, \dots, B_s$. Then the probability of ever observing this witness tree is bounded by
$$
P(\text{$\hat \tau^{t} = \tau$ for some $t \in \mathbf Z$}) \leq P_{\Omega}(B_1) \cdots P_{\Omega} (B_s)
$$

We sometimes refer to the RHS as the \emph{weight} of $\tau$,
$$
w(\tau) = P_{\Omega}(B_1) \cdots P_{\Omega} (B_s)
$$
\end{lemma}
\begin{proof}
The first step of the MT algorithm is to draw all the variables independently from $\Omega$. We may consider fixing the variables to some arbitrary (not random) values, and allowing the MT algorithm to run from that point onward.  We refer to this as \emph{starting at an arbitrary state of the MT algorithm.} We prove by induction on $\tau$ the following: for any witness tree $\tau$, and starting at any state of the MT algorithm, then
we have
\begin{equation}
\label{tt1}
P(\bigvee_{t > 0} \hat \tau^t = \tau) \leq \frac{\prod_{B \in \tau} P_{\Omega}(B)}{\prod_i P_{\Omega}(X_i = A_i(\tau))}
\end{equation}

First, the base case when $\tau = \emptyset$. Then this is vacuously true, as the RHS of (\ref{tt1}) is equal to $1$.

Next, for the induction. Suppose that $\hat \tau^t = \tau$ for some $t > 0$. Then, a necessary condition of this is that, at some time $t' < t$, we must resample some $B \in \mathcal B$ such that a leaf node of $\tau$ is labeled by $B$. (If not, then $\hat \tau^t$ could never acquire such a leaf node.) Suppose that $t'$ is the earliest such time and that $v$ is a leaf node labeled by $B$. We condition now on a specific value for $t'$ and $B$.

For each $i \sim B$, we must resample variable $i$ to take on value $A_i(\tau - v)$ (recall our convention that if $A(\tau - v) = \top$, then this is automatically true.) This has probability $P_{\Omega}(X_i = A_i(\tau-v))$. Next, starting at the state of the system at time $t'+1$, we must satisfy $\hat \tau^{t-1} = \tau - v$. We need to estimate the probability that this occurs, conditional on a fixed choice of $t', B$. Crucially, the inductive hypothesis gives an upper bound on this probability \emph{conditional on any state of the MT algorithm.} In particular, this upper bound applies even when we condition on $t', B$. Thus, we may multiply the two probabilities to obtain:
{\allowdisplaybreaks
\begin{align*}
P(\bigvee_{t > 0} \hat \tau^t = \tau) &\leq \prod_{i \sim B} P_{\Omega} (X_i = A_i(\tau-v)) \frac{\prod_{B' \in \tau-v} P_{\Omega}(B')}{\prod_i P_{\Omega}(X_i = A_i(\tau-v))} \\
&= \frac{\prod_{B' \in \tau} P_{\Omega}(B')}{P_\Omega(B) \prod_{i \not\sim B} P_{\Omega}(X_i = A_i(\tau-v)) } 
\end{align*}
}
If $i \not\sim B$ then $A_i(\tau - v) = A_i(\tau)$, while for each $i \sim B$ we have $(i, A_i(\tau)) \in B$. Hence the denominator is equal to $\prod_i P_{\Omega}(X_i = A_i(\tau))$ and the induction holds.

Next, we claim that the probability that there is some $t > 0$ such $\hat \tau^t = \tau$ is at most $w(\tau)$. (Note that this differs subtly from the induction; in the induction, we are showing a bound which applies to any starting state of the system; here, we are claiming that this bound holds when the MT algorithm begins with a random initialization.) To see this, note that a necessary condition to have $\hat \tau^t = \tau$ is that in the initial sampling of all relevant variables, each variable $i$ must take on value $A_i(\tau)$. Conditional on this event, the probability of $\hat \tau^t = \tau$ for some $t$ is still given by the inductive hypothesis, so we have
$$
P(\bigvee_{t > 0} \hat \tau^t = \tau) \leq \prod_i P_{\Omega}(X_i = A_i(\tau)) \times \frac{\prod_{B \in \tau} P_{\Omega}(B)}{\prod_i P_{\Omega}(X_i = A_i(\tau))} = w(\tau)
$$

\end{proof}

Finally, we show that each event in the execution log of the MT algorithm has a distinct witness tree. This is almost a triviality in the standard analysis of the MT algorithm, but here it is surprisingly subtle. For instance, there may be multiple resamplings of a bad-event $B$, and the later occurrences may have smaller witness trees. Nevertheless, all such trees are unique:
\begin{proposition}
\label{unique-prop}
Let $t_1 < t_2$; then $\hat \tau^{t_1} \neq \hat \tau^{t_2}$.
\end{proposition}
\begin{proof}
Suppose $\hat \tau^{t_1} = \hat \tau^{t_2}$. Proposition~\ref{update-tree-prop} shows that the sequence of trees $\hat \tau_t$ is uniquely determined by the original value $\hat \tau$ and by the sequence of resamplings encountered the execution of the MT algorithm. As $\hat \tau^{t_1} = \hat \tau^{t_2}$ initially, we must have that $\hat \tau_t^{t_1} = \hat \tau_t^{t_2}$ for all $t \geq 1$.
But now substitute $t = t_2$; in this case, $\hat \tau_t^{t_1}$ is the null tree and $\hat \tau_t^{t_2}$ consists of a single node. So this is a contradiction.
\end{proof}

\begin{theorem}
Suppose there is $\mu: \mathcal B \rightarrow [0, \infty)$ satisfying the following condition:
$$
\forall B \in \mathcal B, \mu(B) \geq P_{\Omega}(B) \sum_{\substack{\text{$Y$ orderable}\\\text{to $B$}}} \prod_{B' \in Y} \mu(B')
$$
then the MT terminates with probability 1. The expected number of resamplings of a bad-event $B$ is at most $\mu(B)$.
\end{theorem}
\begin{proof}
First, by induction on tree-height, one can show that the total weight of all witness trees rooted in a bad-event $B$, is at most $\mu(B)$. This follows since the children of the root node form an orderable set for $B$.

Next, by Proposition~\ref{unique-prop}, each resampling of $B$ corresponds to a distinct witness tree. Hence, by the Lemma~\ref{witness-tree-lemma}, the expected number of witness trees rooted in $B$ is at most the sum of the weights of all such trees. Hence the expected number of resamplings of $B$ is at most $\mu(B)$.
\end{proof}

\subsection{Comparison to other LLL criteria}
\label{comparison-sec}
The original form of the LLL simply counted the \emph{number} of neighbors of each bad-event. As noted by \cite{bissacot}, the criterion can be strengthened by further analysis of the dependency graph; namely, for any bad-event $B$, one only needs to examine \emph{independent sets} of neighbors of $B$. Later, Pegden showed that this improved criterion applies also to the MT algorithm \cite{pegden}. Alternatively, following \cite{kolipaka}, one can derive this strengthened criterion as a corollary of \cite{shearer}. This can be stated as follows:

\begin{theorem}[Pegden's Criterion]
\label{pegden-thm}
Suppose there is $\mu: \mathcal B \rightarrow [0, \infty)$ satisfying the following condition:
\begin{align*}
\forall B \in \mathcal B \qquad \mu(B) \geq P_{\Omega}(B) \Bigl[ \mu(B) + \sum_{\substack{\text{$Y$ an independent set}\\\text{of neighbors to $B$}}} \prod_{B' \in Y} \mu(B') \Bigr]
\end{align*}
then the MT terminates with probability 1. The expected number of resamplings of a bad-event $B$ is at most $\mu(B)$.
\end{theorem}

This criterion applies for any dependency graph. In particular, it applies if $\sim$ is defined in terms of lopsidependency ($B \sim B'$ if they disagree) or in terms of simple dependency ($B \sim B'$ if they agree or disagree). One counter-intuitive aspect to Theorem~\ref{pegden-thm} is that sometimes a \emph{denser} dependency graph gives a stronger criterion. (For the Shearer criterion, this can never occur). In particular, ignoring lopsidependency can give better bounds.

Strictly speaking, Pegden's criterion is incomparable to ours. However, in practice, the usual method of accounting for independent sets of neighbors comes from analyzing, for each variable $i$, the total set of all bad-events in which $i$ could participate. An independent set of neighbors of $B$ can contain one or zero bad-events involving each variable. Any higher-order interaction --- such as finding groups of variables participating jointly in bad-events --- is usually too complicated to analyze and is disregarded.

When we account for the dependency graph solely in terms of variable intersection, then we can replace the somewhat confusing concept of ``orderable set'' with a simpler (albeit slightly weaker) notion. 
\begin{Definition}
Given an event $E$, we say that a set of bad-events $Y \subseteq \mathcal B$ is \emph{assignable} to $E$, if either $Y = \{E \}$, \emph{or} there is an \emph{injective} function $f: Y \rightarrow E$, such that for all $B \in Y$, we have some $B \ni z \sim f(B) \in E$
\end{Definition}

\begin{proposition}
 If $Y$ is orderable to $B$, then it is assignable to $B$ (but not necessarily vice-versa)
\end{proposition}

\begin{proof}
Let $Y = \{ B_1, \dots, B_s \}$, so that for each $i = 1, \dots, s$ there is some $z_i \in E$ with 
$$
z_i \sim B_i \qquad z_i \not \sim B_1, \dots, B_{i-1}
$$

Now define $f(B_i) = z_i$.  We claim that $f$ is injective. For, suppose $z_i = z_j$ and $i < j$. Then $z_i \sim B_i$, so $z_j \sim B_i$, which is a contradiction.
\end{proof}

When we sort bad-events by their variables, we obtain the following criteria; these are respectively the LLLL criterion and Pegden's LLL criterion:
\begin{proposition}
\begin{enumerate}
\item If for all bad-events $B$ we have
$$
\mu(B) \geq P_{\Omega}(B) \Bigl[ \mu(B) + \prod_{(i,j) \in B} \prod_{j' \neq j} \prod_{B' \ni (i, j')} ( 1  + \mu(B')) \Bigr]
$$
then the MT algorithm converges.

\item If for all bad-events $B$ we have
$$
\mu(B) \geq P_{\Omega}(B) \prod_{(i,j)\in B} \Bigl( 1 + \sum_{j'} \sum_{B' \ni (i,j')} \mu(B') \Bigr)
$$
\end{enumerate}
then the MT algorithm converges.
\end{proposition}

Our criterion blends these two conditions and is stronger than either of them:
\begin{proposition}
\label{blend-crit}
If for all bad-events $B$ we have
$$
\mu(B) \geq P_{\Omega}(B) \Bigl[ \mu(B) + \prod_{(i,j) \in B} (1 + \sum_{j' \neq j} \sum_{B' \ni (i,j')} \mu(B')) \Bigr]
$$
then the MT algorithm terminates with probability 1; the expected number of resamplings of $B$ is at most $\mu(B)$.
\end{proposition}
\begin{proof}
For the bad-event $B$, we have the criterion
$$
\mu(B) \geq P_{\Omega}(B) \sum_{\substack{\text{$Y$ assignable}\\\text{to $B$}}} \prod_{B' \in Y} \mu(B')
$$
We enumerate the assignable sets $Y$ as follows. First, we may take $Y = \{B \}$; this accounts for the term $\mu(B)$ in the RHS of Proposition~\ref{blend-crit}. Next, for each variable $(i,j) \in B$, we may select either zero or one bad-event $B' \ni (i, j')$ for some $j'  \neq j$. These account for respectively the terms $1$ and $\sum_{j' \neq j} \sum_{B' \ni (i,j')} \mu(B')$ in Proposition~\ref{blend-crit}.
\end{proof}
It is in this sense that we view our criterion as being stronger than the original MT lopsidependency criterion and stronger than Pegden's criterion.

\section{Parallel algorithm for MT}
\label{parallel-mt}
In \cite{moser-tardos}, a simple parallel algorithm was introduced which is based on the MT algorithm:
\begin{enumerate}
\item[1.] Draw all variables from $\Omega$
\item[2.] While there is some true bad-event, repeat the following:
\begin{enumerate}
\item[2a.] Select a maximal independent set $I$ of true bad-events
\item[2b.] Resample, in parallel, all bad-events in $I$.
\end{enumerate}
\end{enumerate}

This algorithm depends on the fact that, in the standard MT framework, bad-events which are unconnected do not share any variables. Hence they do not interact in any way and can be resampled in parallel. This is no longer the case for the lopsidependent MT algorithm; so in that case, frustratingly, we do not have any corresponding parallel algorithms.\footnote{In \cite{det-lll}, there is a brief discussion about a parallel deterministic algorithm for the variable-assignment LLLL. However, no algorithm is provided, nor are there any definite claims made for its performance.}

In this section, we introduce a new parallel algorithm corresponding to the lopsidependent MT setting, which achieves our new criterion up to a multiplicative slack. We assume that each bad-event uses at most $M \leq \text{polylog}(n)$ terms.  We also suppose that the number of bad-events is polynomially bounded, although this can be relaxed quite a bit. Finally, we require a multiplicative slack in the LLLL criterion.

Here is the basic idea. Suppose we have a large number of bad-events which are currently true. Due to the LLLL criterion, there may be many ``unconnected'' bad-events which are simultaneously true, yet they intersect in variables and cannot be resampled in parallel. However, suppose we resample a given variable; with good probability, it will change its value, thereby falsifying \emph{all} of the bad-events which contain it, even those we did not explicitly resample.

This argument can break down if we have $p_{ij} \approx 1$ for any variable $i$ and value $j$ (this situation is rare; see Section~\ref{sec:Ramsey} for an example). In that case, resampling the variable $i$ a single time is not likely to flip its value; we must resample it multiple times. Much of the complication of our parallel algorithm comes from dealing with this somewhat pathological case.  We will begin by stating a simple parallel algorithm which assumes $p_{ij} < 1 - \Omega(1)$ for all $i,j$; we then modify it to remove this condition.
\subsection{The parallel algorithm: warm-up exercise}
We present the following Parallel Moser-Tardos Algorithm (Simplified):
\begin{enumerate}
\item[1.] Draw all variables independently from the distribution $\Omega$.
\item[2.] While there is some true bad-event, repeat the following for rounds $t = 1, 2, \dots, $:
\begin{enumerate}
\item[3.] Let $V_{t,1}$ be the set of bad-events which are true at the beginning of round $t$.
\item[4.] Repeat the following for a series of sub-rounds $s = 1, 2, \dots, $ until $V_{t,s} = \emptyset$.
\begin{enumerate}
\item[5.] Select a maximal disjoint set $I_{t,s} \subseteq V_{t,s}$.  (This can be done using a parallel MIS algorithm).
\item[6.] Resample all $B \in I_{t,s}$.
\item[7.] Update $V_{t,s+1}$ as:  $V_{t,s+1} = V_{t,s} - I_{t,s}  - \text{All bad events which are no longer true}$.
\end{enumerate}
\end{enumerate}
\end{enumerate}

\begin{theorem}
Suppose that we satisfy the condition 
$$
\forall B \in \mathcal B, \mu(B) \geq P_{\Omega}(B) (1 + \epsilon) \sum_{\substack{\text{$Y$ orderable}\\\text{to $B$}}} \prod_{B' \in Y} \mu(B')
$$

Suppose further that we satisfy the condition 
$$
\forall i,j, P_{\Omega}(X_i = j) < 1 - \psi
$$

Let us define
$$
W = \sum_{B \in \mathcal B} \mu(B)
$$

Then whp the Parallel Moser Tardos (Simplified) terminates in time $\psi^{-1} \epsilon^{-1} \log W \log^{O(1)} ( n m )$ using $(n m)^{O(1)}$ processors.
\end{theorem}
\begin{proof}
We provide only a sketch, as we will later introduce a more advanced algorithm. In each round $t,s$, note that every bad-event $B \in V_{t,s}$ contains some resampled variable. Such a variable switches to a new value with probability $\geq \psi$, in which case $B$ is removed from $V_{t,s+1}$. Hence the expected size of $V_{t,s}$ is decreasing as $(1 - \psi)^s$. So, for $s = \Omega(\psi^{-1} \log n)$, we have $V_{t,s} = \emptyset$ and the round $t$ is done.

Next, suppose that a bad-event is resampled in round $t$. One can show that the witness tree for this resampling must have height $t$ exactly. One can also compute the total weight of all such trees, and show that this is decreasing as $(1+\epsilon)^{-t}$. Taking a union-bound over such trees, this implies that the probability of having $\geq t$ rounds is at most $(1+\epsilon)^{-t} W$.

This implies that, whp, our algorithm requires $\psi^{-1} \epsilon^{-1} \log^{O(1)} n \log W $ rounds. In each round, one must select an MIS of the currently-true bad-events, which takes at most $\log^{O(1)} m$ time and $m^{O(1)}$ processors.
\end{proof}

It will take a lot more work to drop the dependency of the running time on $\psi$. To do this, we will need to resample a variable multiple times in the round. Thus, we must replace the step of selecting a maximal disjoint set of bad-events with a maximal set in which no variable occurs too many times (which depends on its probabilities $p_i$). We must also deal with the possibility that we get inconsistent results when we resample a variable multiple times in a round.

Before we move on to the general case, we will need a subroutine to solve a problem we refer to as the \emph{vertex-capacitated maximal edge packing problem} (VCMEP). We believe this may be a useful building block for other parallel algorithms.

\subsection{Vertex-capacitated maximal edge packing}

\begin{Definition}
Suppose we are given a hypergraph $G$, with $m$ edges of size $\leq k$, on a vertex set $V$. For each $v \in V$, we are given a \emph{capacity} $C_v$ in the range $\{0, \dots, m \}$. We wish to select a subset $L \subseteq E$ of the edges, with the property that each vertex appears in at most $C_v$ edges of $L$, and such that $L$ is a maximal subset of $E$ with that property. Such a set $L$ is referred to as a \emph{vertex-capacitated maximal edge packing (VCMEP)}.
\end{Definition}

Such a set can be found easily by a sequential algorithm. Note that if $C_v = 1$ for all $v$, this is equivalent to finding a maximal independent set of the line graph of $G$.

\begin{theorem}
There is parallel algorithm to find a VCMEP in time $k \times \log^{O(1)} (mn)$.
\end{theorem}
\begin{proof}
We will repeatedly add edges until we have reached such a maximal set. At round $i$, suppose we have selected so far edges $L_i$, and we begin with $L_0 = \emptyset$.

Now form the residual graph and residual capacities; we abuse notation so that these are also denoted $G, C$.  One can form an integer program corresponding to the vertex-capacitated \emph{maximum} edge packing (i.e. packing of highest cardinality) for the residual. We let $M_i$ denote the size of the maximum packing which can be obtained by extending $L_i$. This integer program has variables $x_f$ corresponding to each edge $f \in G$, along with constraints that $\sum_{v \in f} x_f \leq C_v$ for each vertex $v$. Now relax the integer program to a positive linear program. As shown by \cite{luby-nisan}, there is a parallel algorithm running in time $\log^{O(1)} (n+m)$ which can find a solution $x'$ which is at least $(1 - \epsilon)$ times the optimum solution, where $\epsilon > 0$ is some sufficiently small constant. In turn, this solution is at least $(1 - \epsilon) (M_i - |L_i|)$.

We now round this fractional solution $x'$ as follows: each edge is selected with probability $x'_f/(2 k)$; if any vertex constraint $v$ is violated, then all edges containing $v$ are de-selected.  We define $L_{i+1}$ to be $L_i$ plus any selected edges.

Define the potential function $\Phi_i = M_i - |L_i|$.  Note that if $\Phi_i = 0$, then $L_i$ must be a maximal set of edges. We claim that, conditional on the state at the beginning of round $i$, we have $\bE[\Phi_{i+1}] \leq (1 - \Omega(1/k)) \Phi_i$.

For, consider some edge $f$; it is selected with probability $x'_f/(2k)$; suppose we condition on that event. Consider any vertex $v \in f$. The expected number of times that other edges incident to $v$ are selected is $\sum_{f' \ni v, f' \neq f} x'_f/(2k) \leq (C_v - x'_f)/(2k)$. By Markov's inequality, the probability that the actual number exceeds $C_v$, in which case $f$ is de-selected, is at most $\frac{C_v - x'_f}{2 k C_v}$. Hence the total probability that $f$ is selected is at least
\begin{align*}
P(\text{$f$ selected}) &\geq \frac{x'_f}{2 k}( 1 - \sum_{v \in f} \frac{C_v - x'_f}{2 k C_v} ) \geq \frac{x'_f}{2 k}( 1 - k \times \frac{1}{2 k}) \geq \frac{x'_f}{4 k}
\end{align*}

Summing over all such edges, we have that 
\begin{align*}
\bE[ \Phi_{i+1} ]&= \bE[ M_{i+1}] - \bE[L_{i+1}] \\
&\leq M_i - |L_i| - \sum_{f} \frac{x'_f}{4 k}  \\
&\leq M_i - |L_i| - (\frac{(1-\epsilon) (|M_i| - L_i)}{4 k})  \\
&\leq \Phi_i (1 - \Omega(1/k))
\end{align*}

Hence, for $i \geq \Omega(k \log (mn))$, we have $\bE[\Phi_i] \leq n^{-\Omega(1)}$. This implies that $\Phi_i = 0$ with high probability, which 
in turn implies that $L_i$ is a maximal packing with high probability.
\end{proof}

\subsection{The parallel algorithm}
\label{full-parallel-section}
We now present our full parallel algorithm. We will suppose that each bad-event uses at most $M$ terms. We also suppose that the number of bad-events is polynomially bounded, although this can be relaxed quite a bit. Finally, we suppose there is a slack in the LLL condition,
$$
\forall B \in \mathcal B, \mu(B) \geq P_{\Omega}(B) (1 + \epsilon) \sum_{\substack{\text{$Y$ orderable}\\\text{to $B$}}} \prod_{B' \in Y} \mu(B')
$$

\begin{enumerate}
\item[1.] Draw all variables independently from the distribution $\Omega$.
\item[2.] While there is some true bad-event, repeat the following for rounds $t = 1, 2, \dots, $:
\begin{enumerate}
\item[3.] Let $V_{t,1}$ be the set of bad-events which are true at the beginning of round $t$. Let $a_i$ be the value of variable $X_i$ at the beginning round $t$. Note that each bad-event in $V_{t,1}$ is a conjunction of terms $X_i = a_i$. For notation throughout the rest of this algorithm, for each variable $i$ let $q_i = P_{\Omega}(X_i \neq a_i)$.
\item[4.] Repeat the following for a series of sub-rounds $s = 1, 2, \dots, $ until $V_{t,s} = \emptyset$.
\begin{enumerate}
\item[5.] View $V_{t,s}$ as a hypergraph, whose vertices correspond to variables and whose hyper-edges correspond to bad-events. For each variable $i$, define the capacity $C_i = \lceil \frac{1}{M q_i} \rceil$. Find a VCMEP $I_{t,s} \subseteq V_{t,s}$.
\item[6.] For each $B \in I_{t,s}$ and each variable $i \sim B$, draw a resampling value $x_{B,i}$ from its distribution in $\Omega$. This represents that \emph{if} we decide to resample $B$, then we will choose to set variable $X_i$ equal to $x_{B,i}$. 
\item[7.] For each $B \in I_{t,s}$ choose a random  $\rho(B)$ independently from the real interval $[0,1]$. We think of $\rho(B)$ as the priority of $B$; we will resample the bad-events in the order of increasing $\rho$. Construct the undirected graph $G_{t,s}$ whose vertices correspond to elements of $I_{t,s}$, and where there is an edge from $B_1$ to $B_2$ if $\rho(B_1) < \rho(B_2)$ and $B_1, B_2$ both share a variable $i$ and we have $x_{B_1,i} \neq a_i$.
\item[8.] Find the lexicographically-first MIS (LFMIS) $I'_{t,s} \subseteq I_{t,s}$ of the graph $G_{t,s}$, with respect to the order $\rho$. (We will say more about this step later)
\item[9.] For each variable $X_i$, if there is some $B \in I'_{t,s}$ with $x_{B,i} \neq a_i$, set $X_i = x_{B,i}$ (by the way that $G_{t,s}$ is constructed, there can be at most one such $B$ for each variable $i$); we say such that variable $i$ is \emph{switched}; otherwise leave $X_i = a_i$. 
\item[10.] Update $V_{t,s+1}$ as $V_{t,s+1} = V_{t,s} - I_{t,s} - \text{All bad events containing a switched variable}$.
\end{enumerate}
\end{enumerate}
\end{enumerate}

This algorithm is quite intricate to analyze. There are two main parts to the proof: showing that the number of rounds is small, and showing that each round can be executed quickly.

\subsection{Bounding the number of rounds}
The key to showing that the number of rounds is small, is to show that this parallel algorithm is simulating a version of the sequential MT algorithm.

\begin{proposition}
\label{witness-tree-lemma2}
Consider the following sequential algorithm, which is a variant of the MT algorithm with an unusual rule for selecting which bad-event to resample:

\begin{enumerate}
\item[1.] Draw all variables independently from the distribution $\Omega$.
\item[2.] While there is some true bad-event, repeat the following for rounds $t = 1, 2, \dots, $:
\begin{enumerate}
\item[3.] Let $V_{t,1}$ be the set of bad-events which are true at the beginning of round $t$. Let $a_i$ be the value of variable $X_i$ at the beginning round $t$.
\item[4.] Repeat the following for a series of sub-rounds $s = 1, 2, \dots, $ until $V_{t,s} = \emptyset$.
\begin{enumerate}
\item[5.] View $V_{t,s}$ as a hypergraph, whose vertices correspond to variables and whose hyper-edges correspond to bad-events. For each variable $i$, define the capacity $C_i = \lceil \frac{1}{M q_i} \rceil$. Find a VCMEP $I_{t,s} \subseteq V_{t,s}$.
\item[6.] Select a random ordering $\pi_{t,s}$ of $I_{t,s}$. For $i = 1, \dots, |I_{t,s}|$ let $\pi_{t,s}(i)$ denote the $i^{\text{th}}$ element of $I_{t,s}$ in this ordering.
\item[7.] For $k = 1, \dots, |I_{t,s}|$ do the following:
\begin{enumerate}
\item[8.] If $\pi_{t,s}(k)$ is currently true, resample it.
\end{enumerate}
\item[9.] Update $V_{t,s+1}$ as $V_{t,s+1} = V_{t,s} - I_{t,s} - \text{All bad events containing a switched variable}$.
\end{enumerate}
\end{enumerate}
\end{enumerate}

Let $X_{i,t,s}$ denote the value of variable $i$ after round $t$ and sub-round $s$. Then the random variables $X_{i,t,s}$ have the same distribution for the parallel algorithm and for this sequential resampling algorithm.
\end{proposition}
\begin{proof}
We will show this by coupling the sequential and parallel algorithms. We consider the following hybrid algorithm which we denote $\mathcal H$.
\begin{enumerate}
\item[1.] Draw all variables independently from the distribution $\Omega$.
\item[2.] While there is some true bad-event, repeat the following for rounds $t = 1, 2, \dots, $:
\begin{enumerate}
\item[3.] Let $V_{t,1}$ be the set of bad-events which are true at the beginning of round $t$.
\item[4.] Repeat the following for a series of sub-rounds $s = 1, 2, \dots, $ until $V_{t,s} = \emptyset$.
\begin{enumerate}
\item[5.] View $V_{t,s}$ as a hypergraph, whose vertices correspond to variables and whose hyper-edges correspond to bad-events. For each variable $i$, define the capacity $C_i = \lceil \frac{1}{M q_i} \rceil$. Find a VCMEP $I_{t,s} \subseteq V_{t,s}$.
\item[6.] For each $B \in I_{t,s}$ and each variable $i \sim B$, draw a random variable $x_{B,i}$ independently from $\Omega$ and draw a random variable $\rho(B)$ independently from $[0,1]$.
\item[7.] Form a permutation $\pi$ of $I_{t,s}$ by sorting in increasing order of $\rho$. 
\item[8.] For $k = 1, \dots, |I_{t,s}|$ do the following:
\begin{enumerate}
\item[9.] If $\pi(k)$ is currently true, then for each variable $i \sim \pi(k)$, set $X_i = x_{\pi(k),i}$.
\end{enumerate}
\item[10.] Update $V_{t,s+1}$ as $V_{t,s+1} = V_{t,s} - I_{t,s} - \text{All bad events containing a switched variable}$.
\end{enumerate}
\end{enumerate}
\end{enumerate}

We first claim that $\mathcal H$ induces the same distribution on the random variables as the sequential algorithm. For, the permutation $\pi$ in $\mathcal H$ is clearly drawn uniformly from the set of permutations of $I_{t,s}$. Also, the algorithm $\mathcal H$ does not examine the variable $x_{\pi(k),i}$ in any way before step (9), and so by the principle of deferred decisions it is equivalent to draw $X_i$ independently from $\Omega$ instead of setting $X_i = x_{\pi(k), i}$.

We next claim that $\mathcal H$ induces the same distribution on the random variables as the parallel algorithm. This follows from the following stronger claim: suppose that we fix the random variables $x_{B,i}$ and $\rho$. Then the value of the variables $X$ is \emph{identical} in $\mathcal H$ and the parallel algorithm. To see this, consider some $B \in I_{t,s}$ with $\pi(k) = B$. Then a simple induction on $k$ shows that $B \in I'_{t,s}$ if and only if $B$ is true at stage $k$ of the loop (9) of $\mathcal H$.

\end{proof}

One may build witness trees for the resamplings of this sequential algorithm, as it is merely a variant of the usual MT algorithm.
\begin{proposition}
\label{height-prop}
Suppose that $B$ is resampled in round $t$. Then the witness tree corresponding to this resampling has height $t$.
\end{proposition}
\begin{proof}
For each $t' \leq t$, let $\hat \tau_{(t')}$ denote the tree formed for the resampling of $B$ from round $t'$ onward (that is, we only add events in rounds $t', \dots, t$ inclusive to the witness tree).

 We will prove by induction the stronger claim: Suppose that $B$ is resampled in round $t$. Then for each $t' \leq t$, the tree $\hat \tau_{(t')}$ has height exactly $t - t' + 1$; furthermore, all the nodes at depth $t - t' + 1$ correspond to bad-events resampled at round $t'$. (Depth 1 corresponds to the root of the tree).

The base case of this induction is $t' = t$. In this case, note that all events resampled in round $t$ agree on all variables, and each bad-event $B \in \mathcal B$ is resampled at most once. Hence $\hat \tau_{(t)}$ consists consists of just a singleton node labeled by $B$.

We move on to the induction step. We begin with $\hat \tau_{(t')}$ and wish to extend it backward in time to round $t'- 1$. By induction hypothesis, $\hat \tau_{(t')}$ has height exactly $t - t' + 1$ and the nodes at depth $t - t' + 1$ correspond to bad-events resampled at round $t'$. 

Note first that all bad-events encountered in round $t' - 1$ are true at the beginning of that round. So they agree on all variables, which implies that they cannot be children of each other. This implies that the only possible nodes at depth $t'- t + 2$ in $\hat \tau_{(t'-1)}$ correspond to bad-events resampled in round $t'-1$ which have as their parent a node of depth $t' - t + 1$.  Thus, the height of $\hat \tau_{(t'-1)}$ is either $t - t' + 2$ (as we want to show), or is $t - t' + 1$.

By induction hypothesis, the tree $\tau_{(t')}$ contains some node $v$ labeled by $B'$ at height $t' - t + 1 $ resampled in round $t'$. 

First suppose $B'$ is true at the beginning of round $t'-1$, so $B' \in V_{t'-1,1}$.  Then either $B'$ is resampled in round $t'-1$, or $B'$ becomes false during round $t'-1$. In the first case, $B'$ would be eligible to be placed as a child of $v$, so it is either placed there or at some other position at the same depth; either way, $\hat \tau_{(t'-1)}$ would have height $t - t' + 2$. In the second case, it must be that $B'$ contains a variable which switched in round $t' - 1$. This implies that $B'$ remains false at the end of round $t' - 1$, so $B' \not \in V_{t',1}$; but $B'$ was resampled in round $t'$ so this is a contradiction.

Second suppose $B'$ is false at the beginning of round $t'-1$. It must have become true due to some variable $X_i$ switching in round $t'-1$ due to resampling some $B''$. But then $B''$ disagrees with $B'$ on variable $X_i$, so $B'' \sim B'$. As $v$ is a leaf node in $\hat \tau_{(t')}$, this implies that $B''$ would be eligible to placed as a child of $v$. Again, such $B''$ will be placed either as a child of $v$ or at some position at the same depth, so that $\hat \tau_{(t'-1)}$ would have height $t - t' + 2$.
\end{proof}

\begin{proposition}
The parallel algorithm terminates after $O(\frac{\log W}{\epsilon})$ rounds whp.
\end{proposition}
\begin{proof}
By Proposition~\ref{witness-tree-lemma2}, it suffices to show that the sequential algorithm terminates after $O(\frac{\log W}{\epsilon})$ rounds whp.  In each round $t$, there is at least one resampling, which must correspond to some tree of height $t$. As shown in \cite{moser-tardos}, due to the slack condition the total weight of all such trees rooted in a bad-event $B$ is $O( \mu(B) (1+\epsilon)^{-t})$. Summing over all such $B$, this implies that for $t = \Omega(\log (n W) )$ this weight is $n^{-\Omega(1)}$. Hence whp no such trees appear.
\end{proof}

\subsection{Analyzing the run-time of individual rounds}
We next consider the individual steps that make up a round of the parallel algorithm. 
\begin{proposition}
The LFMIS $I'_{t,s}$ can be found whp in time $O(\frac{\log n}{\log \log n})$.
\end{proposition}
\begin{proof}
 In general, the problem of finding the LFMIS is P-complete \cite{lfmis-p}, hence we do not expect a generic parallel algorithm for this. However, what saves us it that the ordering $\rho$ and the graph $G_{t,s}$ are constructed in a highly random fashion.  This allows us to use a greedy algorithm to construct $I'_{t,s}$:

\begin{enumerate}
\item[1.] Let $H_1$  be the directed graph obtained by orienting all edges of $G_{t,s}$ in the direction of $\rho_{t,s}$. Repeat the following for $l = 1, 2, \dots,$:
\begin{enumerate}
\item[2.] If $H_l = \emptyset$ terminate.
\item[3.] Find all source nodes of $H_l$. Add these to $I'_{t,s}$.
\item[4.]  Construct $H_{l+1}$ by removing all source nodes and all successors of source nodes from $H_l$.
\end{enumerate}
\end{enumerate}
The output of this algorithm is the LFMIS $I'_{t,s}$. Each step can be implemented in parallel time $O(1)$. The number of iterations of this algorithm is the length of the longest directed path in $G_{t,s}$. So it suffices it show that, whp, all directed paths in $G_{t,s}$ have small length.
 
Suppose we select $B_1, \dots, B_l \in \mathcal B$ uniformly at random. Let us analyze how these could form a directed path in $G$. 

First, note that $B_2$ is a neighbor of $B_1$. Each variable $i \sim B_1$ appears in at most $C_i - 1$ other bad-events. If $B_1$ and $B_2$ intersect in variable $i$, then that variable $i$ creates an edge between $B_1, B_2$ only if $x_{B_1, i} \neq a_i$, which occurs with probability $q_i$.  Thus, for each $B_1$, the expected number of $B_2$ which are connected to that $B_1$ is at most
$$
\sum_{i \sim B_1} (C_i - 1) q_i \leq \sum_{i \sim B_1} 1/(M q_i) \times q_i \leq 1
$$
Continuing this way, we see that the expected number of $B_1, \dots, B_l$ which are connected is at most $1$.

Next, it must be the case that $\rho(B_1) < \rho(B_2) < \dots < \rho(B_l)$. So far, none of the probabilistic statements have referred to $\rho$, so the probability this occurs conditional on all previous events is $1/l!$.
  Thus, for $l \geq \Omega(\frac{\log n}{\log \log n})$, this is $n^{-\Omega(1)}$ as desired.
\end{proof}

\begin{proposition}
For $s = \Omega(M \log n)$, we have $V_{t,s} = \emptyset$ whp.
\end{proposition}
\begin{proof}
We will show that $V_{t,s}$ has an expected size which is decreasing exponentially in $s$.

First, we show the following fact: given any $B \in I_{t,s}$, we have $B \in I'_{t,s}$ with probability $\geq 1/2$. For, a sufficient condition for $B \in I'_{t,s}$ is that there is no variable $i \sim B$ with $B' \in I_{t,s}, \rho(B') < \rho(B)$ and $x_{B',i} \neq a_i$. For each variable $i$, there are at most $C_i - 1$ candidate $B' \in I_{t,s}$, and each of them has probability $q_i$ of setting $x_{B', i} \neq a_i$, so the expected number of such $B'$ is at most $q_i \times \frac{1}{M q_i} \times 1/2 \leq \frac{1}{2 M}$. Over all variables $i \sim B$, this gives a total probability of $\leq 1/2$.

Now consider any $B \in V_{t,s}$. By maximality of $I_{t,s}$, either $B \in I_{t,s}$, or $B$ contains some variable which occurs $C_i$ times in $I_{t,s}$. In the former case, $B$ is necessarily removed from $V_{t,s+1}$.

Now, suppose variable $i \sim B$ occurs exactly $C_i$ times in $I_{t,s}$. For each such occurrence $B'$, there is a probability of $\geq 1/2$ that $B' \in I'_{t,s}$. The event that $B' \in I'_{t,s}$ is independent of $x_{B', i}$, so each such $B'$ has a probability of $q_i/2$ that $B'$ is selected \emph{and} $x_{B',i} \neq a_i$. Hence the total expected number of $B' \in I'_{t,s}$ with $x_{B',i} \neq a_i$, is at least $C_i \times q_i/2 \geq \frac{1}{2 M}$. Note that there are either zero or one elements $B' \in I'_{t,s}$ with this property, hence the probability that $X_i$ switches is at least $\frac{1}{2 M}$. If this occurs, then we have $B \notin V_{t,s+1}$.

In either case, we have shown that a given $B \in V_{t,s}$ is removed with probability at least $1 - \frac{1}{2 M}$. Hence we have
$$
\bE[ | V_{t,s+1} | \mid \text{stage before stage $t,s$}] \leq (1 - \frac{1}{2 M}) |V_{t,s}|
$$
which implies that for $s \geq \Omega(M \log n)$ we have $V_{t,s} = \emptyset$ with high probability.
\end{proof}

Putting this all together, we have the following:
\begin{theorem}
\label{parallel-thm}
Suppose that each $B \in \mathcal B$ has size at most $M$. Suppose that we satisfy the condition 
$$
\forall B \in \mathcal B, \mu(B) \geq P_{\Omega}(B) (1 + \epsilon) \sum_{\substack{\text{$Y$ orderable}\\\text{to $B$}}} \prod_{B' \in Y} \mu(B')
$$

Then whp the Parallel MT algorithm terminates in time $\epsilon^{-1} M (\log W ) (\log^{O(1)} n) (M + \log^{O(1)} m)$ using $(n m)^{O(1)}$ processors.
\end{theorem}

Note that the running time of the Simplified Parallel Moser-Tardos algorithm does not depend on $M$. In practice, when $M$ is large, then other parallel aspects of the Moser-Tardos algorithm can become problematic; for example, we need non-trivial parallel algorithms to enumerate and check the events of $\mathcal B$. It remains an interesting open problem to find a parallel algorithm which works in the regime in which $M$ is large \emph{and} the probabilities of the bad-events become close to $1$.
 
\section{Applications}
\label{app-section}

\subsection{SAT with bounded variable occurrences}
\label{sat-section}
Consider the following problem: we have a SAT instance, in which each clause contains at least $k$ variables. We are also guaranteed that each variable occurs in at most $L$ clauses, either positively or negatively. How can large can $L$ be so as to guarantee the existence of a solution to the SAT instance? This problem was first introduced by \cite{kratochvil}, which showed some bounds on $L$. Most recently, it was addressed by \cite{gst}; they showed that the criterion $L \leq \frac{2^{k+1}}{e (k+1)}$
suffices to guarantee that a solution exists (and can be found efficiently). This criterion is also shown to be asymptotically optimal (up to first-order terms). The main proof for this is to use the LLLL; they show that the worst-case behavior comes when each variable appears in a balanced way (half positive and half negative).

Although the criterion of \cite{gst} is asymptotically optimal, we can still improve its second-order terms. We show the following bound:
\begin{theorem}
If each variable appears at most
$$
L \leq \frac{2^{k+1} (1 - 1/k)^k}{k-1} - \frac{2}{k}
$$
times then the SAT instance is satisfiable, and the MT algorithm finds a satisfying occurrence in polynomial time.

Furthermore, suppose $L \leq \frac{2^{k+1} (1 - 1/k)^k}{(k-1)(1 + \epsilon)} - \frac{2}{k}$.
then whp the Parallel Resampling Algorithm finds a satisfying occurrence in time $\frac{ (k \log n)^{O(1)}}{\epsilon}$.
\end{theorem}
\begin{proof}
We will only prove the sequential result; the parallel result is almost identical.

For each SAT clause, we have a bad-event $B$ that it is violated. We define $\mu(B) = \alpha$ for each bad-event, where $\alpha > 0$ is a constant to be chosen.

As described by \cite{gst}, the key problem is to choose a good probability distribution for each variable. Suppose a variable $i$ occurs in $l_i$ clauses, of which it occurs $\delta_i l_i$ positively. In this case, we set variable $i$ to be T with probability $1/2 - x (\delta_i - 1/2)$, where $x \in [0,1]$ is a parameter to be chosen. This is quite counter-intuitive. One would think that if a variable occurs positively in many clauses, then one should set the variable to be T with high probability; in fact we do the opposite. 

We now wish to show that our MT criterion is satisfied. Let $C$ be a clause and suppose without loss of generality each variable appears in it negatively. Then the corresponding bad-event is that all such variables are true. This has probability $\prod_{i \in C} (1/2 - x (\delta_i - 1/2))$.  Now, consider the assignable sets for the clause. We may either select the singleton $C$ itself, or for each of the $k$ variables we may select one or zero other clauses in which the corresponding variable appears \emph{positively}. For each such variable $i$, the total number of such clauses is at most $\delta_i L$. Hence we have the criterion:
$$
\alpha \geq \prod_{i \in C} \Bigl( 1/2 - x (\delta_i - 1/2) \Bigr) \Bigl(\alpha + \prod_{i \in C} (1 + \delta_i L \alpha) \Bigr)
$$

We bound the RHS as follows:
\begin{equation}
\label{e1}
\begin{aligned}
\prod_{i \in C} \Bigl( 1/2 - x (\delta_i - 1/2) \Bigr) \Bigl(\alpha + \prod_{i \in C} (1 + \delta_i L \alpha) \Bigr) \leq  \prod_{i \in C} \Bigl( 1/2 - x (\delta_i - 1/2) \Bigr) \Bigl( 1 + \delta_i L \alpha + \alpha/k \Bigr)
\end{aligned}
\end{equation}

Now set $x = \frac{\alpha k L}{2 \alpha + 2 k + \alpha k L}$; clearly $x \in [0,1]$. With this choice, verify that that the RHS of (\ref{e1}), viewed as a function of $\delta_i$, achieves its maximum value at $\delta_i = 1/2$. Thus we have
\begin{align*} 
\prod_{i \in C} \Bigl( 1/2 - x (\delta_i - 1/2) \Bigr) \Bigl(\alpha + \prod_{i \in C} (1 + \delta_i L \alpha) \Bigr) \leq \prod_{i \in C} \frac{1}{2}(1 + \alpha/k + \alpha L/2) = 2^{-k} (1 + \alpha/k + \alpha L/2)^k
\end{align*}

We thus need to find $\alpha \geq 0$ such that 
\begin{equation}
\label{sat-e1}
\alpha - 2^{-k} (1 + \alpha/k + \alpha L/2)^k \geq 0
\end{equation}

Differentiate with respect to $\alpha$ to make the LHS of (\ref{sat-e1}) as large as possible. This yields our optimal choice of $\alpha$ namely:
$$
\alpha = \frac{2 k \left( (\frac{2^{k+1}}{2 + k L})^{\frac{1}{k-1}} -1\right)}{2 + k L}
$$

When $L \leq \frac{2^{k+1} (1 - 1/k)^k}{k-1} - \frac{2}{k}$, note that $(\frac{2^{k+1}}{2 + k L})^{\frac{1}{k-1}} \geq \frac{k}{k-1}$ and so $\alpha \geq 0$ as desired. Also, (\ref{sat-e1}) is satisfied. 
\end{proof}

\subsection{Hypergraph coloring}
Suppose we have a $k$-uniform hypergraph, in which each vertex appears in at most $L$ edges. We wish to $c$-color this hypergraph, while avoiding any monochromatic edges. There are many types of graphs and parameters for which better bounds are known, but the LLL gives very simple constructions and also provides the strongest bounds in some cases (particularly when $c, k$ are fixed small integers)\cite{mcdiarmid}.

Let us first examine how the conventional LLL analysis would work. Counter-intuitively, when $c$ is large it is better to use the standard LLL (defining $\sim$ in terms of simple dependency) and when $k$ is large it is better to use the LLLL (defining $\sim$ in terms of lopsidependency.) In the first case, a bad-event is that an edge is monochromatic (of an unspecified color). Consider an edge $f$. The neighbors of $f$ would be other edges that intersect $f$. An independent set of neighbors of $f$ consists of either $f$ itself, or for each vertex $v \in f$ we may select one or zero edges (other than $f$). Setting $\mu(B) = \alpha$ for all bad-events, this gives us the criterion
$$
\alpha \geq c^{1-k} (\alpha + (1 + (L-1) \alpha)^k)
$$
Routine calculations show that this can be satisfied if $L \leq \frac{c^{k-1} (1 - 1/k)^{k-1}}{k}$.

Alternatively, in the LLLL, a bad-event would be that an edge $f$ receives some color $j$. The neighbors of this event would be other edges receiving colors other than $j$; there are $k (L-1) (c-1) + c$ such neighbors. Using the symmetric LLL and some simplifications, one obtains the bound $L \leq \frac{c^k}{(c-1) e k}$. Pegden's criterion could improve this somewhat, although there would no longer be a simple closed form.

There seems to be a ``basic'' form of the bound $L \leq \frac{c^{k-1}}{e k}$; the standard LLL framework can improve on this using either Pegden's criterion
(replacing $1/e$ by $(1 - 1/k)^{k-1}$) or by lopsidependency (replacing one factor of $c$ by $(c-1)$), but cannot do both simultaneously.

\textbf{Our new LLLL criterion.} For each edge $f \in G$, we have $c$ bad-events, namely that $f$ is monochromatic of any given color. We assign $\mu(B) = \alpha$ to all bad-events, where $\alpha > 0$ is a parameter to be determined. We color each vertex independently and uniformly.

Now consider a bad-event $B$, say without loss of generality that edge $f$ receives color $1$. It has probability $c^{-k}$. Consider the orderable sets for $B$; we want to sum $\prod_{B' \in Y} \mu(B')$ over all such sets $Y$.

First, $Y$ may consist of $B$ itself; this contributes $\alpha$. Second,  $Y$ may consist of, for each vertex $v \in f$, zero edges or one edge other than $f$ receiving one color $2, \dots, c$. 
Finally, we may have one vertex select $f$ and some color for it; some set of other vertices selects other edges and other colors. 
Summing all these cases, we have the criterion for $B$:
\begin{equation}
\label{bnd1}
\begin{aligned}
\alpha \geq c^{-k} \Bigl[ (1 + \alpha  (c-1) (L-1))^k  +\alpha  (c-1) ((1 + \alpha  (c-1)
   (L-1))^k  -(\alpha  (c-1) (L-1))^k)+\alpha \Bigr]
\end{aligned}
\end{equation}

This has no closed-form solution for general $L, k$. But, for fixed values of $L, k$ it is easily solvable. For example, when $c = 2$, we list the largest values of $L$ which are obtained by either our improved MT criterion or the original MT criterion (listed under $L'$):
\begin{center}
\begin{tabular}{|c|c|c||c|c|c|} \hline  
$k$ & $L$ & $L'$ & $k$ & $L$ & $L'$ \\
\hline
4 & 2 & 2 & 8 & 13 & 12 \\
5 & 3 & 3 & 9 & 23 & 21 \\
6 & 5 & 4 & 10 & 40 & 38\\
7 & 8 & 7 & 11 & 72 & 69 \\
\hline
\end{tabular}
\end{center}

We see that our new criterion indeed gives (slightly) stronger bounds. For the asymptotic case when $L$ is large, note that the RHS of (\ref{bnd1}) can be approximated:
\begin{align*}
\text{RHS}  &\leq c^{-k} (\alpha + (1 + \alpha  (c-1) (L-1))^k (1 + (c-1) \alpha)) \leq c^{-k} (1 + \alpha  (c-1) L)^k
\end{align*}

Thus, setting $\alpha = \frac{\left(\frac{c^k}{(c-1) k L}\right)^{\frac{1}{k-1}}-1}{(c-1) L}$,
we satisfy the LLLL criterion if
$$
L \leq \frac{c^k (1-1/k)^{k-1}}{(c-1) k}.
$$
which is slightly better than the bounds from the conventional LLLL.

\subsection{Second Hamiltonian cycle}
Consider a $k$-regular graph $G$, with a Hamiltonian cycle $C$. Under what conditions is there a second Hamiltonian cycle $C'$ (that is, the cycle $C$ is not unique)? In \cite{thomassen}, Thomassen showed that a sufficient condition for the existence of the second cycle is a set of vertices $S \subseteq V$ which is simultaneously a dominating set for $G - C$ and an independent set for $C$. Specifically, $S$ must satisfy the following two conditions:
\begin{enumerate}
\item If $v$ and $w$ are adjacent on the cycle $C$, then $v$ and $w$ are not both in $S$.
\item For any vertex $v \in G$, either $v$ is in $S$, or it is connected to a vertex $w \in S$ via some edge $e \notin C$.
\end{enumerate}

Using the LLL, Thomassen then showed that this can always be satisfied as long as $k \geq 73$. This was based on a simple random process, in which each vertex was put into $S$ independently with probability $p$. Using the LLL with a much more sophisticated random process, Haxell showed that this condition can be satisfied as long as $k \geq 23$ \cite{haxell-second}. It was conjectured that this condition could be satisfied as long as $k \geq 5$.

Haxell's proof is quite involved, and our LLLL criterion would offer little benefit for it (as all the bad-events involve many vertices). In \cite{ghandehari}, there was a simple proof using the LLLL that this condition can be satisfied as long as $k \geq 48$. Our LLLL criterion can be used to give another very simple proof under the condition $k \geq 43$. While not as good as Haxell's construction, the proof is far simpler.

\begin{theorem}
If $G$ is a $k$-regular graph for $k \geq 43$ and $C$ is a Hamiltonian cycle of $G$, then there is a $G-C$-dominating, $C$-independent set $S \subseteq V$.
\end{theorem}
\begin{proof}
Each vertex enters into $S$ independently with probability $p$. There are two types of bad-events: for each edge of $C$, there is an event of type A, that the endpoints are both in $S$; for each vertex of $G - C$, there is an event of type B, that $v$ nor its $k-2$ neighbors outside of $C$ are in $S$. We assign $\mu(B) = a$ for all events of the first type, and $\mu(B) = b$ for all events of the second type. Note that events of type A are lopsidependent only with events of type B, and vice versa.

Now consider an event of type A. It has probability $p^2$. There are two vertices in this edge, each of which participates in $k-2$ events of type B. Similarly, an event of type B has probability $(1-p)^{k-1}$, and each of the $k-1$ vertices participates in two events of type A. Hence our LLLL criteria can be stated as
$$
a \geq p^2 (1 + (k-2) b)^2, \quad b \geq (1-p)^{k-1} (1 + 2 a)^{k-1}
$$
Routine calculations show that this is solvable for $k \geq 43$.
\end{proof}

\subsection{Independent transversals}
Suppose we are given a graph $G$, along with a partition of the vertices $V = V_1 \sqcup \dots \sqcup V_k$ in which each class has size exactly $b$. We would like to select one vertex from each class; this is known as a \emph{transversal}. If we select a transversal which is also an independent set, this is known as an \emph{independent transversal}.  Typically, bounds for the existence of an independent transversal are given in terms of the maximum degree $\Delta$.

Haxell showed that when $b \geq 2 \Delta$, an independent transversal exists, and this is the optimal constant \cite{haxell}. However, this result is non-constructive. The best algorithms for finding independent transversals come from the MT algorithm. A simple application of LLL shows that $b \geq 2 e \Delta$ suffices to guarantee an independent transversal. Pegden's criterion shows that $b \geq 4 \Delta$ suffices. We slightly can improve this, obtaining the best constructive bound known so far:
\begin{proposition}
\label{indep-trans-prop}
Suppose $b \geq 4 \Delta - 1$. Then the MT algorithm finds an independent transversal in polynomial expected time. Furthermore, under these conditions, the Parallel MT algorithm runs in time 
$$
\log^{O(1)} n \times \min(1, \frac{4 (b-1) \Delta}{b^2 - 4 (b - 1) \Delta } )
$$
\end{proposition}
\begin{proof}
We prove the first statement only; the second is similar.

Each edge corresponds to a bad-event; it has probability $1/b^2$. For an assignable set of neighbors to an edge $f = \langle u, v \rangle$, we may choose $f$, or we may choose $\langle u', x \rangle$ where $u' \neq u$ is in the class of $u$, or we may choose $\langle v', x \rangle$ where $v' \neq v$ is in the class of $v$; or we may choose both. This gives us the criterion
$$
\alpha \geq b^{-2} (\alpha + (1 + (b-1) \Delta \alpha)^2)
$$
which is satisfied by some $\alpha \geq 0$ whenever $b \geq 4 \Delta - 1$.
\end{proof}

Note that the second condition gives an RNC algorithm either if $b \geq 4 \Delta (1 + \epsilon)$ for some constant $\epsilon > 0$, or 
if $b \geq 4 \Delta - 1$ and $\Delta = \log^{O(1)} n$.

\subsection{Off-diagonal Ramsey numbers}
\label{sec:Ramsey}
In this section, we consider the classical off-diagonal Ramsey problem on graphs. Suppose we wish to two-color -- with colors red and blue -- the edges of $K_n$, the complete graph on $n$ vertices. We wish to avoid any red $s$-cliques or blue $t$-cliques in the resulting graph. The largest value $n$ for which it is possible to avoid such cliques is known as the \emph{off-diagonal Ramsey number $R(s,t)$}. There are many aspects and generalizations studied for Ramsey numbers. One frequently studied scenario is when $s$ is held constant while $t \rightarrow \infty$. 

 It was shown in \cite{spencer-ramsey}, using the LLL, that when $n \leq c (t/\log t)^{\frac{s+1}{2}}$ and $c$ is a constant (depending on $s$) that such a coloring is possible. In other words, $R(s,t) \geq \Omega_s\bigl( (t/\log t)^{\frac{s+1}{2}} \bigr)$. For specific values of $s$, better bounds are known (e.g., $R(3,t) = \Theta(t^2/\log t)$), but this is the best bound known for general $s$.
The algorithmic challenge is to efficiently find colorings of the edges of $K_n$ that avoid red $K_s$ and blue $K_t$. Such algorithms should operate when $n$ is as large as possible, ideally up to $R (s,t)$. Unfortunately, the LLL construction of \cite{spencer-ramsey} does not lead to efficient serial or parallel algorithms. The main roadblock is that there is a bad-event for each $t$-clique, so that finding a bad event requires exponential time. For specific values of $s$, again, there are known polynomial-time algorithms for finding good colorings. But in general there is no algorithm that corresponds to the best bounds.

In \cite{hss}, an algorithm based on MT was proposed for finding such colorings. The basic idea of \cite{hss} is to find and resample red $K_s$, while ignoring blue $K_t$. One then shows that, high probability, none of the $K_t$ became blue, even though we did not explicitly check or resample them. The serial running time of this would be $\Omega_s (n^s)$, to search for the red $K_s$; no parallel algorithm was given. 

These results depend on the ``MT-distribution''; that is, the distribution on the variables when the MT algorithm terminates. We will show that in the MT distribution there is only a small probability of a blue $K_t$. Though we did not show this explicitly in Section~\ref{new-mt-proof}, it is not hard to use Lemma~\ref{witness-tree-lemma} to derive a bound on the MT distribution, similar to \cite{hss}:
\begin{theorem}
\label{mt-dist-thm}
Suppose we have a set of bad-events $\mathcal B$ which satisfies our LLLL criterion which weights $\mu$. Suppose $E$ is any atomic event (which is not itself in $\mathcal B$). Then the probability that $E$ is true at the end of the MT algorithm is given by
\begin{align*}
P(\text{$E$ is true at end of MT}) \leq P_{\Omega}(E) \sum_{\text{$Y$ orderable to $E$}} \prod_{B' \in Y} \mu(B')
\end{align*}
\end{theorem}
\begin{proof}
One can construct a witness tree for the first time that $E$ becomes true. Such a witness tree is rooted in $E$, and its children correspond to a set of bad-events which is orderable to $E$. Thus, the result follows by taking a union-bound over all such witness trees. See \cite{hss} for more details.
\end{proof}

Using this result,we will give a serial algorithm for off-diagonal Ramsey coloring with a much better run-time, and we will also give a parallel algorithm based on our Parallel Resampling Algorithm.
\begin{theorem}
\label{thm:off-diagonal-ramsey}

Consider the problem of finding a red-blue coloring of the edges of $K_n$ avoiding red $K_s$ and blue $K_t$.  Define 
$$
c_s = 
   \left(\frac{2}{s}-\frac{2}{s-1}+1\right)^{\frac{s+1}{2}}  \left( \frac{2 (s-2)!}{s (s-1)^{\binom{s}{2}}} \right)^{\frac{1}{s-2}}.
$$

\begin{enumerate}
\item Suppose $n \leq (\frac{t}{\log t})^{\frac{s+1}{2}} (c_s - o(1))$. There is a serial randomized algorithm which runs in time $n^{s/4 + O(1)}$ and produces a correct solution when it halts, except with a failure probability of $n^{-\Omega(1)}$.
\item Suppose $s$ is constant and $n \leq (\frac{t}{\log t})^{\frac{s+1}{2}} (c_s - o(1))$. There is a parallel randomized algorithm which runs in time $s^{O(1)} \log^{O(1)} n$ time using $n^{s/4+O(1)}$ processors, and produces a correct solution when it halts, except with a failure probability of $n^{-\Omega(1)}$.
\end{enumerate}
\end{theorem}
\begin{proof} 
The proofs of both parts are very similar; to simplify the discussion, we will focus mostly on the serial algorithm, noting any difference between that and the parallel algorithm.

The probability space $\Omega$ is defined by coloring each edge red with probability
$$
p = \left(\frac{2 (s-2)!}{(s-1) s} \right)^{\frac{2}{s^2-s-2}} n^{\frac{-2}{s+1}} 
$$
and blue otherwise. 
We ignore the blue $K_t$ and so our only bad-events are the red $K_s$. Each bad-event has probability $q = p^{\binom{s}{2}}$. Observe that each bad-event is lopsidependent with only a single bad-event, namely itself. So the LLLL criterion is satisfied, setting $\mu(B) = \frac{q}{1 - q}$ for all bad-events $B$.

Now consider an arbitrary $K_t$, and let $E$ be the event that it is red at the end of the MT. We have $P_{\Omega}(E) = (1-p)^{\binom{t}{2}}$. The orderable sets for this event can be found as follows: for each of the $\binom{t}{2}$ edges, we may select zero or one blue $K_s$. Thus, by Theorem~\ref{mt-dist-thm}, the probability that $E$ holds at the end of MT is given by 
{\allowdisplaybreaks
\begin{align*}
P(\text{$K_t$ is blue}) &\leq \Bigl(  (1-p) (1 +  \binom{n-2}{s-2} \mu) \Bigr)^{\binom{t}{2}} \\
& \leq \Bigl(  1 - \left( \frac{s (s-1)}{2 (s-2)!} \right)^{\frac{2}{-s^2+s+2}} \left( 1 + \frac{2}{s-s^2} \right) n^{\frac{-2}{s+1}} \Bigr)^{\binom{t}{2}} \\
& \leq \exp\Bigl( -\binom{t}{2} c'_s n^{\frac{-2}{s+1}} \Bigr)  \quad \quad \text{where $c'_s = \Bigl( \frac{2 (s-2)!}{s (s-1)} \Bigr)^{\frac{2}{s^2-s-2}} \Bigl( 1 + \frac{2}{s-s^2} \Bigr) $}
\end{align*}
}

Hence, the expected number of blue $K_t$ is at most 
$$
\bE[\text{Blue $K_t$}] \leq \frac{n^t}{t!}  \exp(-t^2 c'_s n^{\frac{-2}{s+1}}/2)
$$

In order to avoid all blue $K_t$ with high probability, say with probability $n^{-\phi}$ for $\phi > 0$ some arbitrary constant, we must take
\begin{align*}
n &\leq \left(\frac{c'_s t^2}{(s+1)(\phi+t)
   \log \left(\frac{c'_s t^2 (t!)^{-\frac{2}{(s+1)
   (\phi +t)}}}{(s+1) (\phi +t)}\right)}\right)^{\frac{s+1}{2}}
=\bigl( t/\log t \bigr)^{\frac{s+1}{2}} ( c_s - o_s (1) )
\end{align*}

So far, we have shown that when we run the MT algorithm with the given parameters, then indeed we avoid $K_t$ with high probability. The next thing we must examine is how to run the MT algorithm. In this problem, in which the bad-events are defined to be red $K_s$, the MT algorithm is somewhat degenerate. The critical thing to note is that when we resample a bad-event, we can never create new bad-events. Thus, the most potentially time-consuming step of MT --- repeatedly searching for any bad-events which are currently true --- can be much simplified. At the beginning of the process, after we make the initial random color assignment but before we do any resamplings, we can enumerate all red $K_s$. For each such red $K_s$, we repeatedly sample the edges until the $K_s$ is no longer red. The process of finding the red $K_s$ can be aided by the fact that we are searching for them in a \emph{random} graph --- namely, the edges are red independently with probability $p$.

The simplest way to search for such red $K_s$ seems to be through a branching process: we gradually build up red $K_{l}$, for $l \leq s$. In this branching process, the expected number of red $K_l$ is $p^{\binom{l}{2}} \binom{n}{l}$. Thus, the total time complexity of this branching process will be
\begin{align*}
\text{Time to find red $K_s$} 
&\leq n^{O(1)} \sum_{l = 0}^s \binom{n}{l} p^{\binom{l}{2}} \\
&\leq n^{O(1)} \exp( \max_{l \in [0,s]} l \log n + \frac{(l-1)^2}{2} \log p - l \log l + l ) \\
& \leq n^{s/4 + O(1)} \qquad \text{(by routine calculus)}
\end{align*}

This concludes part (1) of the theorem. As the Parallel Moser-Tardos algorithm simulates the serial algorithm, all the results about the MT distribution still remain true in the parallel setting; in particular we avoid the blue $K_t$ with high probability. Also, one can enumerate the red $K_s$ in parallel, using $n^{s/4 + O(1)}$ processors and $s$ stages, via the same type of branching process. 

 One can see easily that we satisfy the parallel LLLL criterion for $n$ sufficiently large, setting $\mu(B) = 1$ for all $B$ and $\epsilon = 1/2$. All the bad events then use $\binom{s}{2}$ elements, and the total number of bad-events is $n^{O_s(1)}$, and we have $\log \sum_{B \in \mathcal B} \mu(B) = O_s(\log n)$. Hence by Theorem~\ref{parallel-thm}, the Parallel Moser-Tardos terminates in time $s^{O(1)} \log^{O(1) } n$ whp.
\end{proof}

\section{Acknowledgments}
Thanks to Aravind Srinivasan for many helpful comments and suggestions, as well as suggesting the algorithm for vertex-capacitated maximal edge covering. Thanks to the anonymous referees for their suggestions.


\begin{thebibliography}{1}
\bibitem{harris3}
Harris, D.: The Moser-Tardos lopsidependency criterion can be stronger than Shearer's criterion. arxiv
\end{thebibliography}
\end{document}